%% file: jcss.tex
\documentclass[11pt,a4paper]{article}

\usepackage{fancybox}
\usepackage{fullpage}
\usepackage{color,soul}
\usepackage{hyperref}
\usepackage{amssymb, amsmath, amsthm, mathtools}
\usepackage{dutchcal}
\usepackage{todonotes}
\usepackage{pst-tree}

\usepackage{enumerate}
\usepackage{soul}

\theoremstyle{plain}
\newtheorem{theorem}{Theorem}
\newtheorem{lemma}[theorem]{Lemma}
\newtheorem{corollary}[theorem]{Corollary}

\theoremstyle{definition}
\newtheorem{definition}{Definition}

\allowdisplaybreaks

\title{\bf On approximate pure Nash equilibria in weighted congestion games with polynomial latencies}
\author{Ioannis Caragiannis\thanks{Department of Computer Science, Aarhus University, 
{\AA}bogade 34, 8200 Aarhus N, Denmark. Email: \texttt{iannis@cs.au.dk}.} \and Angelo Fanelli\thanks{CNRS, (UMR-6211), France. Email: \texttt{angelo.fanelli@unicaen.fr}.}}

\def\RP{\mathbb{R}^{\geq 0}} % positive reals
\def\RPP{\mathbb{R}^{> 0}} % strictly positive reals
\def\RPPP{\mathbb{R}^{\geq 1}} %  reals larger than 1
\def\Z{\mathbb{Z}^{\geq 1}} % strictly positive integers

\def\WCGD{{{WCG}$(\d)$ }} % game
\def\WCG{\mathrm{WCG}} % game (math mode or arg different from d)

\def\N{{N}} % set of players
\def\E{{E}} % set of resources
\def\w{{w}} % demand
\def\S{{S}} % set of strategies
\def\a{{a}} % coefficient
\def\k{{k}} % exponent
\def\l{{\ell}} % latency
\def\d{{d}} % degree
\def\x{{x}} % variable 1
\def\h{{h}} % parameter
\def\up{\texttt{up}} % LB strategy up 
\def\down{\texttt{down}} % LB strategy down 

\def\ss{\mathbf{s}} % generic state
\def\ee{\mathbf{e}} % equilibrium
\def\oo{\mathbf{o}} % optimum
\def\SoS{ {S}} % set of states
\def\s{s} % strategy

\def\P{{P}} % subset of players
\def\L{{L}} % set of players using a resource

\def\c{{c}} % cost function
\def\C{{C}} % social cost

\def\POT{{\Phi}} % global potential function
\def\POTONE{{\Phi}_{\mathbf{1}}} % global potential function 1
\def\Pot{{\Psi}} % local potential function

\def\OPT{{\textrm{OPT}}} % set of optima

\newcommand{\game}[1]{\langle{#1}\rangle} 

\newcommand{\PoS}[1]{\mathrm{PoS}_{#1}}

\newcommand{\EQ}[1]{\mathcal{E}_{#1}}

\newcommand{\coeff}[1]{\gamma_{#1}}

\newcommand{\add}[1]{\textcolor{black}{#1}}

\date{}

\begin{document}

\maketitle

%TODO mandatory: add short abstract of the document
\begin{abstract}
We consider weighted congestion games with polynomial latency functions of maximum degree $d\geq 1$. For these games, we investigate the existence and efficiency of approximate pure Nash equilibria  which are obtained through sequences of unilateral improvement moves by the players.

By exploiting a simple technique, we firstly show that these games admit an infinite set of $\d$-approximate potential functions. This implies  that there always exists a $\d$-approximate pure Nash equilibrium which can be reached through any sequence of $\d$-approximate improvement moves by the players. As a corollary, we also obtain that, under mild assumptions on the structure of the players' strategies, these games also admit a constant approximate potential function. Secondly,  using a simple potential function argument, we are able to show that a $(\d+\delta)$-approximate pure Nash equilibrium of cost at most $(\d+1)/(\d+\delta)$ times the cost of an optimal state always exists, for every $\delta\in [0,1]$. 
\end{abstract}

%%%%%%%%%%%%%%%%%%%%%_INTRO_%%%%%%%%%%%%%%%%%%%%%%%%%%

\section{Introduction}\label{sec:intro}
\input{introduction}

%%%%%%%%%%%%%%%%%%%%%_FUNCTIONS_%%%%%%%%%%%%%%%%%%%%%%%%%%

%\section{A class of state-valued functions}\label{sec:func}
%\input{functions}

%%%%%%%%%%%%%%%%%%%%%_POTENTIALS_%%%%%%%%%%%%%%%%%%%%%%%%%%

\section{Approximate potential functions}\label{sec:approx}
\input{potentials}

%%%%%%%%%%%%%%%%%%%%%_POS_%%%%%%%%%%%%%%%%%%%%%%%%%%

\section{Approximate price of stability}\label{sec:pos}
\input{pos}

%%%%%%%%%%%%%%%%%%%%%_PRELIM_%%%%%%%%%%%%%%%%%%%%%%%%%%

%\section{Technical lemmas}\label{sec:tech}
%\input{preliminaries}

%%%%%%%%%%%%%%%%%%%%%_CONCLUSIONS_%%%%%%%%%%%%%%%%%%%%%%%%%%

\section{Conclusion}\label{sec:discussion}
Our work leaves several open questions. For example, can our techniques yield a better than $d$-approximate potential function? After several unsuccessful attempts to define better potentials, we conjecture that this is not possible. In particular, we believe that for any better than $d$-approximate potential function defined as in Lemma~\ref{thm:fundamental}, there exists a game in \WCGD in which the range condition (\ref{thm:fundamental:EQ}) of Lemma~\ref{thm:fundamental} is violated (at some resource). Unfortunately, we have been unable to prove such a statement so far.

Still, it is important to further expand the class of approximate potential functions for weighted congestion games with polynomial latencies as they may have several applications; let us mention a few here. First, we believe that approximate potential functions will be useful in bounding the convergence time to states with nearly-optimal social cost, extending the results of Awerburch et al.~\cite{AAE+08}, who focused on games that admit exact potential funtions. Second, it is worth investigating whether our results can be combined with the approach in \cite{CaragiannisFGS11,CFGS15} to compute approximate equilibria in polynomial time; the use of our new approximate potential functions could replace the Faulhaber potential in the analysis of \cite{GNS18} and simplify it or even improve it. Third, new approximate potential functions will be useful in determining the best possible bounds on the approximate price of stability. Our bounds in Theorem~\ref{thm:pos} are very close to $1$ but, still, they are not know to be tight (for $\delta\in [0,1)$). Finally, what about the $\rho$-approximate price of stability for $\rho<d$?

\subsection*{Acknowledgments}
A preliminary version of the results in this paper appeared in Proceedings of the 46th International Colloquium on Automata, Languages and Programming (ICALP), pages 133:1-12, 2019. The work was partially supported by the French Project ANR-14-CE24-0007-01 CoCoRICo-CoDec and by COST Action 16228 GAMENET (European Network for Game Theory).

%%%%%%%%%%%%%%%%%%%%%_BIB_%%%%%%%%%%%%%%%%%%%%%%%%%%

%\newpage
\bibliography{apx-weighted}
\bibliographystyle{plain}

\end{document}

%% file: introduction.tex
%In \cite{CaragiannisFGS15},
%In \cite{HansknechtKS14},
%In \cite{ChristodoulouGGS18}\\

Among other solution concepts, the notion of the pure Nash equilibrium plays a central role in Game Theory. 
Pure Nash equilibria in a game characterize situations in which no player has an incentive to unilaterally deviate from the current situation in order to achieve a higher payoff. 
Unfortunately, it is well known that there are games that do not have any pure Nash equilibrium. 
Furthermore, even in games where the existence of pure Nash equilibria is guaranteed, these equilibria could be very inefficient compared to % ic: replaced "with respect to" with "compared to"
solutions dictated by a central authority. %computation can be a computationally hard task. 
Such negative results %significantly 
question the importance of pure Nash equilibria as solution concepts that characterize the behavior of rational players.

One way to overcome the limitations of the non-existence and inefficiency of pure Nash equilibria is to consider a relaxation of the 
%ANG_stability 
equilibrium condition. 
%Positive results can be achieved by relaxing the strong stability constraints of the pure Nash equilibrium.
This relaxation leads to the concept of \emph{approximate} pure Nash equilibrium;
%ANG_This concept 
it characterizes situations where no player can {\em significantly improve} her payoff by unilaterally deviating from her current strategy.
Approximate pure Nash equilibria can accommodate small modeling inaccuracies (e.g., see the 
discussion
%arguments 
in \cite{CGC04}), therefore they may be more desirable as solution concepts in practical %decision-making 
settings.
Besides 
%ANG_their 
existence and efficiency, 
%existence and efficiency matters,
approximate pure Nash equilibria are an appealing alternative 
also
%solution concepts 
from a computational point of view (e.g., \cite{BiloFMM15, CaragiannisFGS11, CFGS15, ChienS11, FeldottoGKS17}).

In this work, we investigate the existence and efficiency of approximate pure Nash equilibria in  weighted congestion games, \add{with the additional requirement that such equilibria are obtained through  sequences of unilateral improvements by the players}. 
This class of games forms a general framework which models situations where a group of agents compete for the use of a set of shared resources.
In the following, we define weighted congestion games and give a formal statement of the two problems we address. 
We continue this section with a discussion of the related literature and a detailed presentation of our contribution. This discussion includes formal definitions that are necessary in the presentation of our problem statement and the comparison of our results to related work.
\bigskip

\noindent{\bf Weighted congestion games.}
%Typical examples of games in this class are weighted congestion games in networks, where each network link corresponds to a resource, and each player has alternative paths that connect two nodes of the network as strategies. 
%...In a weighted congestion game, players compete over a set of resources. 
A typical example of a subclass of games in this family is that of weighted congestion games in networks, where each player has alternative paths that connect two nodes of the network as strategies to select. 
Each player selfishly chooses a path trying to optimize a local objective. 
Thus, each network link corresponds to a resource that can be shared among the players in the network.
%Each player has a positive weight. 
Each link incurs a latency to all players who use it; this latency depends on the total weight (congestion) of the players that use the link according to a resource-specific, non-negative, and non-decreasing latency function. 
Therefore, among the given set of paths, 
%(over sets of resources), 
each player aims to select one that 
%selfishly, trying to 
minimizes her individual total cost, i.e., the sum of the latencies on the links in her path. 
Departing from games in networks, %in which the set of strategies of a player is the set of paths connecting two nodes,  
we can generalize the set of strategies of a player to any combinatorial structure;
in weighted congestion games, we generally assume that the set of strategies of a player is a subset of the power set of a ground set of resources.  
A particular setting is when latencies are polynomial functions of the total weights of the users; we refer to this setting as weighted congestion games with polynomial latencies. 
In the following we give a formal description of the games in this class. %This game is formally described as follows. 
%Now, let us describe the game more formally. 

\vspace{.3cm}
\fbox{
\parbox{0.9\textwidth}{
%The {\sc weighted congestion game} with polynomial latencies of degree at most $d \in \Z$ is a collection of instances, denoted by \WCGD, of the form $\G = \game{\N, \E, (\w_i)_{i\in \N}, (\S_i)_{i\in \N}, (\a_e, \k_e)_{e\in \E}},$ where $N=\{1,2,\ldots, |\N|\}$ is the set of \emph{players}, $\E=\{1,2,\ldots, |\E|\}$ is the set of \emph{resources}, $\w_i \in \RPP$ is the \emph{weight} of player $i$, $\S_i \subseteq 2^{\E}$ is the set of \emph{strategies} of player $i$ and 
%$(\a_e, \k_e) \in \RPP \times \Z$ $(\a_e, \k_e) \in \RPP \times \{1,2,\ldots, d\}$ are the \emph{coefficient} and the \emph{degree} of resource $e\in \E$ respectively, which encode the \emph{latency function} $\l_e: 2^\N \mapsto \RP$ associated with $e$,  mapping every subset of players $\P \subseteq \N$ to the non-negative real  $a_e\left(\sum_{j\in \P}\w_j\right)^{\k_e}$.
A {\sc weighted congestion game} with polynomial latencies of maximum degree $d \geq 1$ 
%*Ang
(we use \WCGD to denote the class of these games) 
consists of a set of \emph{players} $N=\{1,2,\ldots, |\N|\}$ and a set of \emph{resources} $\E=\{1,2,\ldots, |\E|\}$. Each player $i$ is associated with a \emph{weight} $\w_i \in \RPP$ and a \add{non-empty set of} \emph{strategies}  $\S_i \subseteq 2^{\E}$. Every resource $e$ is described by a pair $(\a_e, \k_e) \in \RPP \times \{1,2,\ldots,\d\}$, which encodes the \emph{latency function} $\l_e: 2^\N \mapsto \RP$ associated with $e$,  mapping every subset of players $\P \subseteq \N$ to the non-negative real; specifically this value is given by the product of $\a_e$ with the $\k_e$-th power of the total weights of the players in $\P$, i.e.,  
\[
\l_e(\P) = \a_e\left(\sum_{j\in \P}\w_j\right)^{\k_e}\!\!\!\!\!.
\] 
We refer to $\a_e$ and $\k_e$ as the \emph{coefficient} and the \emph{degree} of resource $e$, respectively, \add{and assume that $\d = \max_{e\in\E}\k_e$}. 
The set of \emph{states} of the game is denoted by $\SoS=S_1 \times S_2 \times \ldots \times S_{|N|}$. 
For every state $\ss\in \SoS$, we refer to its $i$-th component, that is the strategy played  by player $i$ in $\ss$, by $\ss(i)$.
For every state $\ss$ and resource $e$, we denote by $\L_e(\ss)$ the set of players using resource $e$ in $\ss$, i.e., $\L_e(\ss) = \{j\in N : e\in \ss(j)\}$. 
We refer to the sum of the weights of all players in $\L_e(\ss)$ as the \emph{congestion} of $e$ in $\ss$.
For every state $\ss$, the \emph{cost} incurred by player $i$ in $\ss$ is  
\[
\c_i(\ss) = \sum_{e\in \ss(i)} \l_e(\L_e(\ss)).
\]
\add{Notice that, by definition, $\c_i(\ss) > 0$ for every player $i$ and state $\ss$.}
For $\tau>0$, we say that the game is $\tau$-{\em congested} if for every resource $e\in E$, every state $\ss\in \SoS$, and every player $i\in \L_e(\ss)$, it holds that $\sum_{j\in \L_e(\ss)\setminus \{i\}}\w_j \geq \tau \k_e \w_i$.
}
}
\vspace{0.6cm}

\noindent{\bf Equilibria, potential functions, and the price of stability.}
We now introduce concepts that are necessary to formally state our problems and present our results.
%Let us consider an instance $\G = \game{\N, \E, (\w_i)_{i\in \N}, (\S_i)_{i\in \N}, (\a_e, \k_e)_{e\in \E}}$ of \WCGD. 
For every state $\ss \in \SoS$, every player $i\in\N$ and every $s\in \S_i$, we denote by $[\ss_{-i}, s]$ the new state obtained from $\ss$ by setting the $i$-th component, that is the strategy of $i$, to $\s$ and keeping all the remaining components unchanged, i.e., $[\ss_{-i}, s](i) = \s$ and $[\ss_{-i}, s](j) = \ss(j)$ for every player $j\neq i$.

The transition from $\ss$ to $[\ss_{-i}, s]$ is called a \emph{move} of player $i$ from state $\ss$. %, and we say that $[\ss_{-i}, s]$ is obtained by a move of player $i$ from $\ss$.
For $\alpha \geq 1$, we say that a transition from $\ss$ to $[\ss_{-i}, s]$ is an \emph{$\alpha$-improvement move} for $i$ if $\alpha\c_i([\ss_{-i}, s]) < \c_i(\ss)$. 
%(it is a \emph{strictly} $\alpha$-improvement if $\alpha\c_i([\ss_{-i}, s]) < \c_i(\ss)$).
For $\alpha \geq 1$, we say that a state-valued function $\Gamma : \SoS \mapsto \RP$ is an \emph{$\alpha$-approximate potential function} for the game if it strictly decreases at every  $\alpha$-improvement move, i.e., $\Gamma([\ss_{-i}, s]) < \Gamma(\ss)$ whenever $\alpha\c_i([\ss_{-i}, s]) < \c_i(\ss)$. %\add{if $\alpha = 1$ we simply refer to $\Gamma$ as an \emph{exact potential function} rather than a $1$-approximate potential function.} 
If the game admits an $\alpha$-approximate potential function $\Gamma$, then every sequence of $\alpha$-improvement moves leads to a \emph{local optimum} of $\Gamma$, i.e., to a state in which no further $\alpha$-improvement move can be performed. Such a state is called $\alpha$-approximate pure Nash equilibrium. 
Formally, for $\alpha \geq 1$, we say that a state $\ss\in \SoS$ is an \emph{$\alpha$-approximate pure Nash equilibrium} if, for every player $i\in \N$ and every strategy $\s\in\S_i$, we have $\c_i(\ss) \leq \alpha\c_i([\ss_{-i}, s])$; 
if $\alpha = 1$, we simply refer to $\ss$ as an \emph{exact pure Nash equilibrium}, or simply \emph{pure Nash equilibrium}, rather than a $1$-approximate pure Nash equilibrium. 
\add{For any $\alpha \geq 1$ such that the game admits an $\alpha$-approximate potential function}, we denote by $\EQ{\alpha} \subseteq \SoS$ the set of all $\alpha$-approximate pure Nash equilibria of the game.

The \emph{social cost} of state $\ss\in \SoS$ is the weighted sum of the players' costs, i.e., $\C(\ss) = \sum_{i\in \N} \w_i\c_i(\ss)$. Notice that, by summing over the resources instead of the players, $\C(\ss)$ can be rewritten as $\C(\ss) = \sum_{e\in \E} \a_e\left(\sum_{j\in \L_e(\ss)}\w_j\right)^{\k_e + 1}$. 
Every state $\ss\in \SoS$ that minimizes the social cost is called a \emph{social optimum}; we denote by $\OPT$ the set of social optima of the game, i.e.,  $\OPT = \arg\min_{\ss \in \SoS}\C(\ss)$.
For any $\oo \in \OPT$ 
%be any social optimum of the game with positive social cost, 
\add{ and any $\alpha\geq 1$ such that $\EQ{\alpha}\neq \emptyset$}, we define the \emph{$\alpha$-approximate price of stability} of the game as $\PoS{\alpha} = \min_{\ee \in \EQ{\alpha}}\frac{\C(\ee)}{\C(\oo)}$.  
\bigskip

%ic: not necessary in my opinion
%Notice that, if the social cost of any optimum is null then this state would be also a pure Nash equilibrium; this implies that, in this case, there is no loss of efficiency at the best equilibrium and therefore we may assume $\PoS{} = 1$. \\%\alert{($\C(\oo) > 0$)}\todo{todo}\\

\noindent{\bf Problem statements.}
We consider the following two problems for games in \WCGD. 
% in weighted congestion games with polynomial latencies: 
%of degree at most $\d \geq 1$.
%We formally state such problems as follows. 
%In this subsection we formally state the two main problems we address. 
%Informally, in this work we address the problem of 

\vspace{.3cm}
\fbox{
\parbox{0.9\textwidth}{
%\vspace{-.3cm}
\begin{enumerate}[(I)]
	\item {\sc Existence of convergent sequences of $\alpha$-improvement moves.}	In this problem, 
	%, given any instance $\G$ of \WCGD, 
	we seek for a reasonably small $\alpha\geq 1$ for which any sequence of $\alpha$-improvement moves converges to an $\alpha$-approximate pure Nash equilibrium. This would be equivalent to saying that the game admits an $\alpha$-approximate potential function, whose value decreases at every $\alpha$-improvement move and whose local optima coincide with  $\alpha$-approximate pure Nash equilibria.\label{prob:1}  
	\item {\sc Bounding the approximate price of stability.} In this problem, for any value of $\alpha\geq 1$ for which the game admits an $\alpha$-approximate pure Nash equilibrium, %,  given any instance $\G$ of \WCGD, which admits an $\alpha$-approximate pure Nash equilibrium, 
	we aim at bounding its $\alpha$-approximate price of stability.\label{prob:2}
\end{enumerate}
}
}
\vspace{.3cm}

\add{Notice that problem (\ref{prob:1}) is more stringent than the problem of establishing whether the game admits  an approximate pure Nash equilibrium. Here, we also require that such an equilibrium can be achieved through a sequence of $\alpha$-improvement moves.}
%\vspace{.6cm}
\bigskip

\noindent{\bf Related work.}
The unweighted setting (i.e., when all players have unit weights) with general latencies has been widely studied in the literature. 
For this case, Rosenthal \cite{R73} proved that there exists a $1$-approximate potential function;
% with the following remarkable property: the difference in the potential value between two states (i.e., two snapshots of strategies) that differ in the strategy of a single player equals to the difference of the cost experienced by this player in these two states. 
 this immediately implies that every sequence of $1$-improvement moves by the players leads to a pure Nash equilibrium. 
 Unfortunately, this nice property does not carry over when players have weights.
 %Any sequence of improvement moves by the players strictly decreases the value of the potential and a state corresponding to a local minimum of the potential will eventually be reached; this corresponds to a pure Nash equilibrium. 
In fact, \add{if there are at least three players and we restrict to the set of twice continuously differentiable latency functions}, a $1$-approximate potential function exists only when the latencies are linear or exponential \cite{FKS05,HK10,PS06}. For polynomial latencies of constant maximum degree strictly higher than $1$, pure Nash equilibria may not exist \cite{FKS05,Goemans05,Lavy01}.
More generally, for arbitrary \add{non-decreasing} latencies, the problem of deciding whether a given instance  has a pure Nash equilibrium is {NP}-hard \cite{DS08}.
%{\sf NP}-hard \cite{DS08}.
 
 Caragiannis et al. \cite{CFGS15} proved that 
 %*Ang
 any game in \WCGD  
 %a weighted congestion game with polynomial latencies of degree at most $\d$
 admits a $\d!$-approximate potential function.
 This result has been subsequently improved considerably by Hansknecht et al. \cite{Hansknecht14}, who showed that 
 %the game 
  %*Ang
 any game in \WCGD  
 admits a $(\d+1)$-approximate potential function.
 %ic: do we really need to describe the way Hansknecht et al define their potential functions? We do not do so for Rosenthal's potential.
 %The potential function they proposed is a Rosenthal-like potential function. 
%Roughly speaking, they obtained an approximate potential as follows. 
%For each resource, they chose an appropriate fixed ordering of the players. 
%Then, for each resource separately, they computed a discrete integral.
%Specifically, they sum up the latency of the resource  after introducing the first player multiplied with the  weight of the first player, the latency after introducing the first two players multiplied with the weight of the second player, and so on, i.e., $\w_1\l_e(\{\w_1\}) + \w_2\l_e(\{\w_1 + \w_2\}) + \ldots$. 
%The potential obtained depends on the way the players have been initially ordered. 
%The authors showed that, the potential function providing the best approximation for polynomial latencies, that is $\d+1$, is the one obtained by ordering the players in non-decreasing order in terms of their weights, i.e., $\w_1 \leq \w_2 \leq \ldots$.
Their potential function is defined in a parameterized way, using an ordering of the players as parameter. The particular $(d+1)$-approximate potential function is obtained by ordering the players in terms of their weights.
%\add{The approximation guarantee provided by all the functions proposed in the literature is assessed by comparing the marginal contribution of the potential function with the cost incurred by each player; this technique has been used for the first time in \cite{CFGS15} and it is summarized in Lemma \ref{thm:fundamental}.}

The $1$-approximate price of stability for 
%*Ang
games in \WCGD 
%a weighted congestion game with polynomial latencies of degree at most $\d$
has been recently investigated by Christodoulou et al.~in \cite{Christodoulou19}; they provided a lower bound of $\Omega\left((\d/\log d)^{\d+1}\right)$, matching the upper bound of Aland et al.~\cite{Aland11}.
The authors of \cite{Christodoulou19} also showed bounds on the $\alpha$-approximate price of stability. 
Specifically, they proved that 
%*Ang
any game in \WCGD 
%a weighted congestion game with polynomial latencies of degree at most $\d$
with weights ranging in $[1,\w_{\max}]$ has an $\alpha$-approximate pure Nash equilibrium, for any $\alpha$ in the range $\left[\frac{2(\d+1)\w_{\max}}{2\w_{\max} + \d + 1}, \d+1\right]$, whose cost is at most $1 + (\frac{\d+1}{\alpha} - 1)\w_{\max}$ times the cost of any optimal state.  
Their proof exploits a potential function called Faulhaber's potential.
For the unweighted setting, tight bounds of $1.577$ for linear latencies \cite{Caragiannis10, Christodoulou05} and of $\Theta(\d)$ for polynomial latencies of degree $d\geq 1$ \cite{Christodoulou15} are known.
\bigskip

\noindent{\bf Our contribution.}
%\comment{The major contribution of this work is the proposal of a simple class of function (Def 1) which is clearly, as shown by our results, an effective tool to tackle both problem \ref{prob:1} and \ref{prob:2}.In fact, as explained in more details later, this class of function contains approximate potential functions which both provide good approximation guarantee and turns out to be an essential tool to show a constant bound on the approximate price of stability.}
Concerning problem (\ref{prob:1}), we show (in Theorem \ref{thm:approx}) that 
%*Ang
games in \WCGD admit \add{an infinite set of}
$\d$-approximate potential functions.
% \add{and, more generally, $(\d + \delta)$-approximate potential functions, for every $\delta \in [0,1]$}.
This implies that every sequence of $\d$-improvement moves by the players always leads %the game 
to a $\d$-approximate pure Nash equilibrium. 
This result is achieved using the technique that is formalized in Lemma~\ref{thm:fundamental} and the class of state-valued functions $\POT_\gamma$, \add{parametrized by $\gamma$}, defined in Definition \ref{def:POT}.
Essentially, while Definition \ref{def:POT} provides a simple interesting class  of candidate potential functions, Lemma~\ref{thm:fundamental} gives a local condition for each resource to determine the approximation guarantee achieved by a given state-valued function. 
So, by exploiting Lemma~\ref{thm:fundamental}, in Theorem \ref{thm:approx} we are able to show that the class $\POT_\gamma$ contains \add{a large subclass of} 
%\delete{$\d$-approximate potential functions and, more generally,} 
$(\d + \delta)$-approximate potential functions, for every $\delta \in [0,1]$.
% \add{these differ by the value of the parameter $\gamma$}.
%\add{According to how we set $\gamma$ we could obtain ...}

We remark that, our potential functions are substantially different from the potential function proposed in \cite{Hansknecht14}.
In particular, the potential in \cite{Hansknecht14} is obtained in a Rosenthal-like fashion, by ordering the players and summing their costs assuming that each player is affected only by the congestion caused by preceding players in the ordering. In contrast, our potential is much simpler and is obtained by a suitable scaling of the coefficients of the polynomials in the definition of the latency functions.
As a matter of fact, we define potentials which, despite their simplicity, provide an approximation factor that approaches $\d$ (instead of $\d+1$) from below, although it is worth noticing that, for small values of $\d$ (e.g., $d\in \{2,3,4\}$), the approximation shown in \cite{Hansknecht14} coincides with the one guaranteed by $\POTONE$ (see our discussion in Section~\ref{sec:approx} and Table 1 in \cite{Hansknecht14}). 
%(see Table \ref{tab:approx}).
As a corollary of part (b) of Theorem \ref{thm:approx}, we show (in Corollary \ref{cor:opt}) that the social optimum of 
%*Ang
%any game in \WCGD 
the game 
is a $(\d+1)$-approximate pure Nash equilibrium.
\add{This result has been shown before in \cite{Christodoulou19}; our aim with restating it here is to  highlight the merit of the class of state-valued functions $\POT_\gamma$ as an effective tool to effortlessly prove an important property of congestion games}.
More importantly, \add{as an exclusive property of our class of potential functions}, the proof of Theorem \ref{thm:approx} implies, as stated by Corollary \ref{cor:congested}, 
that 
%*Ang
%$\tau$-congested games in \WCGD admit 
$\tau$-congested games admit
approximate potential functions with considerably better approximation guarantees (approaching $1$ for very high congestion). 
%\st{every mildly congested instance of \WCGD  always admits a  $\frac{\mathrm{\bf e}}{\mathrm{\bf e}-1}$-approximate potential function, where $\mathrm{\bf e}$ is the Euler's number.} 
%\alert{rewrite this.}\todo{todo}
%\alert{every  $\beta$-congested instance of \WCGD  always admits a  $...$-approximate potential function, where $\mathrm{\bf e}$ is the Euler's number.}

%We also show that, 
The class of functions $\POT_\gamma$ serves also as an essential tool to give an answer to problem (\ref{prob:2}). %In particular, through $\POT_\gamma$, we are able to show a very low constant bound on the approximate price of stability. 
More specifically, by exploiting \add{ the collection of $(\d + \delta)$-approximate potentials in the class of functions} $\POT_\gamma$, we are able to show (Theorem \ref{thm:pos}) an upper bound of $(\d+1)/(\d+\delta)$ for the $(\d + \delta)$-approximate price of stability, for every $\delta\in [0,1]$. 
To prove this bound, we use the standard potential function argument.
Specifically, we first bound (Lemma \ref{lem:pos}) the value of any $(\d + \delta)$-approximate potential function for a given state in terms of the social cost of that state;
if we then perform a sequence of $(\d + \delta)$-improvement moves starting from an optimal state, the potential does not increase, and hence we can bound the cost of any $(\d + \delta)$-approximate pure Nash equilibrium that we reach. 
Notice that our bound does not depend on the range of the players' weights and significantly improves the bound provided in \cite{Christodoulou19}, by making use of a different and much simpler potential function.
\bigskip

\noindent{\bf Roadmap.} The rest of the paper is structured as follows.
%We begin with the definition of a class of state-valued functions in Section \ref{sec:func}. 
In Section \ref{sec:approx} we first present a simple technique to bound the approximation guarantee of a given state-valued function. Subsequently, we show that a class of state-valued functions provide a good approximation. Then, the bound on the approximate price of stability is presented in Section \ref{sec:pos}. 
We conclude with a discussion on open problems and possible extensions of our results in Section~\ref{sec:discussion}.

%%%%____COMPLEXITY

%% file: potentials.tex
The main result of this section is given by Theorem \ref{thm:approx}, 
%\delete{which states the existence of  approximate potential functions with low approximation factors}
\add{which identifies a class of $\d$-approximate potential functions}.
Before presenting this result, in Lemma~\ref{thm:fundamental} we illustrate the tool we use to design an approximate potential function; 
%\add{to determine the approximation guaranteed by a state-value function}. 
this tool gives a local condition to each resource to determine the approximation guarantee of a given state-valued function. 
%\comment{Lemma 1 gives a sufficient condition in order for a state-valued function to be considered an $\alpha$-approximate potential function for a ...}
We conclude the section with two corollaries. The first (Corollary \ref{cor:opt}) states that the social optimum 
%*Ang
%of an instance of \WCGD 
of the game
is always a $(\d+1)$-approximate pure Nash equilibrium. 
The second (Corollary \ref{cor:congested}) indicates that, under mild conditions, 
%*Ang
%games in \WCGD 
the game
always admits a constant approximate potential function.

\begin{lemma}\label{thm:fundamental}
	%Let  $\G = \game{\N, \E, (\w_i)_{i\in \N}, (\S_i)_{i\in \N}, (\a_e, \k_e)_{e\in \E}}$ be an instance of  \WCGD.
Let  $\Gamma: \SoS\mapsto \RPP$ be a state-valued function such that $\Gamma(\ss) = \sum_{e\in \E}\a_e\Gamma_e\big(\L_e(\ss)\big)$, where $\Gamma_e: 2^\N \mapsto \RPP$. If, for every resource $e\in \E$, every non-empty subset of players $\P\subseteq \N$ and every player $i \in \P$, there exist $\lambda, \upsilon\in \RPP$, with $\lambda \leq \upsilon$, such that 
\begin{equation}\label{thm:fundamental:EQ}
\frac{\w_i\l_e(\P)}{\a_e\Big(\Gamma_e(\P) - \Gamma_e(\P\setminus \{i\})\Big)} 
\in [\lambda, \upsilon]
\end{equation}
then $\Gamma$ is a $\left(\frac{\upsilon}{\lambda}\right)$-approximate potential function.
\end{lemma}

\begin{proof}
	Let us consider a state $\ss\in \SoS$ and a player $i$.
%We have
%\begin{equation*}
%	\w_i\c_i(\ss) 
%= 
%	\sum_{e\in \ss(i)}\w_i\l_e\big(\L_e(\ss)\big)
%\leq
%	\sum_{e\in \ss(i)}\upsilon_e\a_e\Big(\Gamma_e\big(\L_e(\ss)\big) - \Gamma_e\big(\L_e(\ss)\setminus \{i\}\big)\Big),
%\end{equation*}
%where the last inequality follows from \eqref{thm:fundamental:EQ}.
%\begin{equation*}
%	\w_i\c_i([\ss_{-i}, s])
%=
%	\sum_{e\in \s}\w_i\l_e\big(\L_e([\ss_{-i}, s])\big)
%\geq
%	\sum_{e\in \s}\lambda_e\a_e\Big(\Gamma_e\big(\L_e([\ss_{-i}, s])\big) - \Gamma_e\big(\L_e([\ss_{-i}, s])\setminus \{i\}\big)\Big),
%\end{equation*}
%Let $\alpha = \upsilon/\lambda$.
Let us assume that $i$ can perform an $\frac{\upsilon}{\lambda}$-improvement move by replacing strategy $\ss(i)$ with $\s\neq \ss(i)$, i.e., $\frac{\upsilon}{\lambda}\c_i([\ss_{-i}, s]) < \c_i(\ss)$.
In order to prove the claim we need to show that $\Gamma([\ss_{-i}, s]) < \Gamma(\ss)$.
To this aim, let us bound the expression $\Gamma([\ss_{-i}, s]) - \Gamma(\ss)$. 
By the definition of the state-valued function $\Gamma$, we have
%let $\s \in \S_i$ be any strategy of $i$ different from $\ss(i)$, we have
\begin{align}\nonumber
	\Gamma([\ss_{-i}, s]) - \Gamma(\ss)
	&=
	\sum_{e\in \E}\a_e\Gamma_e\big(\L_e([\ss_{-i}, s])\big) - \sum_{e\in \E}\a_e\Gamma_e\big(\L_e(\ss)\big)\\\nonumber
	&=
	\sum_{e\in \E}\a_e\Big(\Gamma_e\big(\L_e([\ss_{-i}, s])\big) - \Gamma_e\big(\L_e(\ss)\big)\Big)\\\nonumber
	&=
	\sum_{e\in \s\setminus \ss(i)}\a_e\Big(\Gamma_e\big(\L_e([\ss_{-i}, s])\big) - \Gamma_e\big(\L_e(\ss)\big)\Big)
	\\ &\hspace{2cm} 
	- \sum_{e\in \ss(i)\setminus \s}\a_e\Big(\Gamma_e\big(\L_e(\ss)\big) -\Gamma_e\big(\L_e([\ss_{-i}, s])\big)\Big).\label{eq:gamma}
\end{align}
Now, observe that player $i$ belongs to $\L_e([\ss_{-i}, s])$ but not to $\L_e(\ss)$ when $e\in s\setminus \ss(i)$ (thus, $\L_e(\ss)=\L_e([\ss_{-i},s])\setminus \{i\}$), while she belongs to $\L_e(\ss)$ but not to $\L_e([\ss_{-i}, s])$ when $e\in \ss(i) \setminus s$ (thus, $\L_e([\ss_{-i},s])=\L_e(\ss)\setminus \{i\}$ in this case). Hence, by applying \eqref{thm:fundamental:EQ} to the right-hand side of \eqref{eq:gamma}, we obtain
\begin{align}	\nonumber
	\Gamma([\ss_{-i}, s]) - \Gamma(\ss)	&\leq
\frac{w_i}{\lambda}\sum_{e\in s\setminus \ss(i)}\l_e\big(\L_e([\ss_{-i}, s])\big) - \frac{w_i}{\upsilon}\sum_{e\in \ss(i)\setminus s}\l_e\big(\L_e(\ss)\big)\\\nonumber
&\leq
\frac{w_i}{\lambda}\sum_{e\in s\setminus \ss(i)}\l_e\big(\L_e([\ss_{-i}, s])\big) - \frac{w_i}{\upsilon}\sum_{e\in \ss(i)\setminus s}\l_e\big(\L_e(\ss)\big)
%\\\nonumber &\phantom{==} 
+\left(\frac{w_i}{\lambda}-\frac{w_i}{\upsilon}\right) \sum_{e\in \ss(i)\cap  s}\l_e\big(\L_e(\ss)\big)\\\nonumber
&= \frac{w_i}{\lambda}\sum_{e\in s}\l_e\big(\L_e([\ss_{-i}, s])\big) - \frac{w_i}{\upsilon}\sum_{e\in \ss(i)}\l_e\big(\L_e(\ss)\big)\\\label{eq:fundamental-last}
	&=\frac{\w_i}{\upsilon}\Big(\frac{\upsilon}{\lambda}\c_i([\ss_{-i}, s]) - \c_i(\ss)\Big).
\end{align}
The second inequality is due to the fact $\upsilon\geq \lambda$. The first equality follows since $\L_e(\ss)=\L_e([\ss_{-i}, s])$ for $e\in \ss(i)\cap  s$ and the last equality follows by the definition of the players' cost. The lemma follows since \eqref{eq:fundamental-last} implies that $\Gamma([\ss_{-i}, s]) < \Gamma(\ss)$ whenever $\frac{\upsilon}{\lambda}\c_i([\ss_{-i}, s]) < \c_i(\ss)$.
\end{proof}

In order to state the main result of this section (Theorem~\ref{thm:approx}), we define a class of state-valued functions mapping every state of the game to a non-negative real number. This class of functions will be exploited in subsequent results as well.

\vspace{.3cm}
\fbox{
	\parbox{0.9\textwidth}{
\begin{definition}\label{def:POT}
	%Let  $\G = \game{\N, \E, (\w_i)_{i\in \N}, (\S_i)_{i\in \N}, (\a_e, \k_e)_{e\in \E}}$ be an instance of  \WCGD.
	For every $\gamma = (\coeff{e})_{e\in\E}$, we define 
	\begin{equation*}
	\POT_\gamma(\ss) = \sum_{e\in \E}\a_e \Pot_e^{\coeff{e}}\big(\L_e(\ss)\big),
	\end{equation*}
	where, for every resource $e\in \E$, it is $\Pot_e^{\coeff{e}}\big(\emptyset\big)=0$ and 
	\begin{equation*}
	\Pot_e^{\coeff{e}}(\P) = \frac{\coeff{e}}{\k_e + 1}\left(\sum_{j\in \P}\w_j\right)^{k_e + 1} + \left(1 - \frac{\coeff{e}}{\k_e + 1}\right)\sum_{j\in \P}\w_j^{\k_e + 1},
	\end{equation*}
	for every nonempty subset of players $\P \subseteq \N$.
	%\add{The function obtained by setting $\gamma_e = 1$ for every $e\in \E$, is denoted  by $\POTONE$.}
\end{definition} 
}}
\vspace{0.3cm}

%\subsection{A class of $\d$-approximate potential functions}

%We are ready to present the main result of this section.
\begin{theorem}\label{thm:approx}
Let $\gamma=(\gamma_e)_{e\in E}$ with $1\leq \coeff{e}\leq k_e+1$ for $e\in E$ and $\gamma^*=\max_{e\in E}{\gamma_e}$. We have 
\begin{itemize}
	\item[(a)] If $\gamma^*=1$ then 
	%If $\gamma_e=1$ for every $e\in\E$ then 
	$\POT_{\gamma}$ (denoted by $\POTONE$ in this case) is a $\rho$--approximate potential function with	
	\begin{equation*}%\label{eq:better-bound}
			\rho = \max_{e\in \E}\sup_{\x\geq 0}\; \frac{\big(1  + \x\big)^{\k_e}}
			{\frac{1}{\k_e + 1}\big(1 + \x\big)^{\k_e + 1} + \frac{\k_e}{\k_e + 1} - \frac{1}{\k_e + 1}\x^{\k_e + 1}} \leq \d.
	\end{equation*}
	%\add{We denote this function by $\POTONE$.}
	\item[(b)] Otherwise (if $\gamma^*>1$), $\POT_{\gamma}$ is a $\max\{\gamma^*,d\}$--approximate potential function.
\end{itemize}
	%Let  $\G = \game{\N, \E, (\w_i)_{i\in \N}, (\S_i)_{i\in \N}, (\a_e, \k_e)_{e\in \E}}$ be an instance of  \WCGD. 
%It holds that
%\begin{enumerate}[(a)]
%	\item if  $\coeff{e}=1$, for every $e\in \E$,  then $\POT_\gamma$ is a $\rho$-approximate potential function, where
%	\begin{equation}
%		\rho := \max_{e\in \E}\sup_{\x>0}\; \frac{\big(1  + \x\big)^{\k_e}}
%		{\frac{1}{\k_e + 1}\big(1 + \x\big)^{\k_e + 1} + \frac{\k_e}{\k_e + 1} - \frac{1}{\k_e + 1}\x^{\k_e + 1}} \;\;\leq \d;
%	\end{equation}
%	\emph{(see Table \ref{tab:approx})};
%	\label{thm:approx:A}
%	\item if  $\coeff{e}=\k_e$, for every $e\in \E$, then $\POT_\gamma$ is a $\d$-approximate potential function;\label{thm:approx:B}
%	\item if $\coeff{e}  = \k_e + \delta$, for every $e\in \E$ and  $\delta \geq 0$, then $\POT_\gamma$ is a $(\d + \delta)$-approximate potential function.\label{thm:approx:C}
%\end{enumerate}
\end{theorem}

\begin{proof}
We prove the claim using Lemma \ref{thm:fundamental}. For every resource $e\in \E$, every non-empty subset of players $\P\subseteq \N$ and every player $i \in \P$, we bound the ratio
\begin{equation}\label{thm:approx:ratio}
	\frac{\w_i\l_e(\P)}{\a_e\Big(\Pot_e^{\coeff{e}}(\P) - \Pot_e^{\coeff{e}}(\P \setminus \{i\})\Big)}.
\end{equation}
For every player $i\in \P$, let $\mu_i(\P) = \frac{1}{\w_i}\sum_{j\in \P \setminus\{i\}} \w_j$. We have,
\begin{align}\label{thm:approx:ratio:numerator-2}
	\w_i\l_e(\P) &= \w_i a_e\left(\sum_{j\in \P}\w_j\right)^{\k_e} = \w_i\a_e\left(\w_ i  + \sum_{j\in \P\setminus\{i\}}\w_j\right)^{\k_e}
	%\nonumber\\ &
		= \a_e\w_i^{\k_e + 1} \big(1  + \mu_i(\P)\big)^{\k_e}.  
\end{align}
Now, let us focus on the expression $\Pot_e^{\coeff{e}}(\P) - \Pot_e^{\coeff{e}}(\P \setminus \{i\})$. 
Using  Definition \ref{def:POT}, we have 
\begin{align}
		 %\Pot_{e,i}^{\coeff{e}}(\P) 
		 &\phantom{==} \Pot_e^{\coeff{e}}(\P) - \Pot_e^{\coeff{e}}(\P \setminus \{i\})\nonumber\\
		 &= \frac{\coeff{e}}{\k_e + 1}\left(\sum_{j\in \P}\w_j\right)^{\k_e + 1} + \left(1 - \frac{\coeff{e}}{\k_e + 1}\right)\sum_{j\in \P}\w_j^{\k_e + 1}\nonumber\\
		 &\phantom{=} - \frac{\coeff{e}}{\k_e + 1}\left(\sum_{j\in \P\setminus \{i\}}\w_j\right)^{\k_e + 1} - \left(1 - \frac{\coeff{e}}{\k_e + 1}\right)\sum_{j\in \P\setminus\{i\}}\w_j^{\k_e + 1}\nonumber\\
		 &= \frac{\coeff{e}}{\k_e + 1}\left(\sum_{j\in \P}\w_j\right)^{\k_e + 1} + \left(1 - \frac{\coeff{e}}{\k_e + 1}\right)\w_i^{\k_e + 1} - \frac{\coeff{e}}{\k_e + 1}\left(\sum_{j\in \P\setminus \{i\}}\w_j\right)^{\k_e + 1}\nonumber\\
		 &= \frac{\coeff{e}}{\k_e + 1}\left(\w_i + \sum_{j\in \P\setminus\{i\}}\w_j\right)^{\k_e + 1} + \left(1 - \frac{\coeff{e}}{\k_e + 1}\right)\w_i^{\k_e + 1} - \frac{\coeff{e}}{\k_e + 1}\left(\sum_{j\in \P\setminus \{i\}}\w_j\right)^{\k_e + 1}\nonumber\\
		 &= \frac{\coeff{e}}{\k_e + 1}\w_i^{\k_e + 1}\big(1 + \mu_i(\P)\big)^{\k_e + 1} + \left(1 - \frac{\coeff{e}}{\k_e + 1}\right)\w_i^{\k_e + 1} - \frac{\coeff{e}}{\k_e + 1}\w_i^{\k_e + 1}\mu_i(\P)^{\k_e + 1}\nonumber\\
		 &= \w_i^{\k_e + 1}\Bigg[ \frac{\coeff{e}}{\k_e + 1}\big(1 + \mu_i(\P)\big)^{\k_e + 1} + \left(1 - \frac{\coeff{e}}{\k_e + 1}\right) - \frac{\coeff{e}}{\k_e + 1}\mu_i(\P)^{\k_e + 1}\Bigg].\label{thm:approx:ratio:denominator-2}
\end{align}
Using \eqref{thm:approx:ratio:numerator-2} and \eqref{thm:approx:ratio:denominator-2}, \eqref{thm:approx:ratio} can be rewritten as
\begin{align}
		&\frac{\w_i\l_e(\P)}
		{\a_e\Big(\Pot_e^{\coeff{e}}(\P) - \Pot_e^{\coeff{e}}(\P \setminus \{i\})\Big)}\nonumber\\
		&= \frac{\a_e\w_i^{\k_e + 1} \big(1  + \mu_i(\P)\big)^{\k_e}}
		{\a_e\w_i^{\k_e + 1}\Bigg[ \frac{\coeff{e}}{\k_e + 1}\big(1 + \mu_i(\P)\big)^{\k_e + 1} + \left(1 - \frac{\coeff{e}}{\k_e + 1}\right) - \frac{\coeff{e}}{\k_e + 1}\mu_i(\P)^{\k_e + 1}\Bigg]}\nonumber\\
		&= \frac{\big(1  + \mu_i(\P)\big)^{\k_e}}
		{\frac{\coeff{e}}{\k_e + 1}\big(1 + \mu_i(\P)\big)^{\k_e + 1} + \left(1 - \frac{\coeff{e}}{\k_e + 1}\right) - \frac{\coeff{e}}{\k_e + 1}\mu_i(\P)^{\k_e + 1}}.\label{thm:approx:res-2}
\end{align}

To proceed, we need the following technical lemma.

%In order to prove part (b) we need the following technical lemma.
\begin{lemma}\label{lem:technical-1}
	For every $\x \in \RP$, $\h \in \Z$ and $\beta \in \RPPP$, we have
	\begin{equation*}
	\frac{(1+\x)^\h}{\frac{\beta}{\h+1}(1+\x)^{\h+1} + (1-\frac{\beta}{\h+1}) - \frac{\beta}{\h+1}\x^{\h+1}} \in \left[\frac{1}{\beta},\max\left\{1,\frac{h}{\beta}\right\}\right].
	\end{equation*}
\end{lemma}
\begin{proof} 
We have
	\begin{align*}
	\frac{(1+\x)^\h}{\frac{\beta}{\h+1}(1+\x)^{\h+1} + (1-\frac{\beta}{\h+1}) - \frac{\beta}{\h+1}\x^{\h+1}}
	&= \frac{\sum_{t=0}^{\h}{\h\choose t}\x^t}{\frac{\beta}{\h+1}\sum_{t=0}^{\h+1}{{\h+1}\choose t}\x^t + (1-\frac{\beta}{\h+1}) - \frac{\beta}{\h+1}\x^{\h+1}}\\  
	%&= \frac{1 + \sum_{t=1}^{\h}{\h\choose t}\x^t}{1 + \beta\frac{1}{\h+1}\sum_{t=1}^{\h+1}{{\h+1}\choose t}\x^t - \beta\frac{1}{\h+1}\x^{\h+1}}\\ 
	&= \frac{1 + \sum_{t=1}^{\h}{\h\choose t}\x^t}{1 + \frac{\beta}{\h+1}\sum_{t=1}^{\h}{{\h+1}\choose t}\x^t}\\
	%&= \frac{1 + \sum_{t=1}^{\h}{\h\choose t}\x^t}{1 + \sum_{t=1}^{\h}\beta\frac{1}{\h+1}\frac{\h+1}{\h+1-t} {{\h} \choose t}\x^t}
	&= \frac{1 + \sum_{t=1}^{\h}{\h\choose t}\x^t}{1 + \sum_{t=1}^{\h}\frac{\beta}{\h+1-t} {{\h} \choose t}\x^t}.
	%&= \frac{
	%1\cdot\x^0 + \phantom{\beta\frac{1}{\h}} {{\h} \choose 1}\x^1 + \phantom{\beta\frac{1}{\h-1}} {{\h} \choose 2}\x^2 + \ldots + \phantom{\beta} {{\h} \choose \h}\x^\h
	%}
	%{
	%1\cdot\x^0  + \beta\frac{1}{\h} {{\h} \choose 1}\x^1 + \beta\frac{1}{\h-1} {{\h} \choose 2}\x^2 + \ldots + %\beta {{\h} \choose \h}\x^\h
	%}
	\end{align*}
	%where \eqref{lem:eq1} holds because 
	%\begin{equation*}
	%	{{\h+1} \choose t} 
	%= \frac{\frac{(\h + 1)!}{t!(\h + 1 - t)!}}{\frac{\h!}{t!(\h - t)!}} {{\h} \choose t} 
	%= \frac{(\h + 1)!}{t!(\h + 1 - t)!}\frac{t!(\h - t)!}{\h!}  {{\h} \choose t}
	%= \frac{\h+1}{\h+1-t} {{\h} \choose t}.
	%\end{equation*}
	%Notice that the term $\frac{\h}{\h+1-t}$ in the summation in \eqref{lem:eq1} gets values in the range $[1, \h]$. 
	The lemma now follows by observing that 
	$$\min\{1,\beta/h\} \left(1 + \sum_{t=1}^{\h}{{\h} \choose t}\x^t\right) \leq 1 + \sum_{t=1}^{\h}\frac{\beta}{\h+1-t} {{\h} \choose t}\x^t\leq \beta \left(1 + \sum_{t=1}^{\h}{{\h} \choose t}\x^t\right).$$
	%In order to bound \eqref{lem:eq2}, for every $t\in [0,\h]$ we consider the ratio between the coefficient of the term $\x^t$ in the numerator and the coefficient of the same term in the denominator.
	%For $t=0$ the ratio is $1$, while for $t \in [1,\h]$ the ratio is $\frac{\h+1-t}{\beta}$.
	%For the case $\beta\in [1,\h]$, we get that the smallest ratio is $1/\beta$ while the greatest is $\h/\beta$.
	%It follows that, when $\beta\in [1,\h]$, the expression in \eqref{lem:eq2} is at least $1/\beta$ and at most $\h/\beta$.
	%For the case $\beta\geq \h$, we obtain that the smallest is $1/\beta$ while the greatest ratio is $1$df.
	%Therefore, when $\beta\geq \h$, the expression in \eqref{lem:eq2} is at least  $1/\beta$ and at most $1$.
	%From which the claim follows.
\end{proof}

By applying Lemma \ref{lem:technical-1} to \eqref{thm:approx:res-2} with $\x = \mu_i(\P)$, $\h = \k_e$ and $\beta = \coeff{e}$, we obtain that
\begin{equation}\label{thm:approx:res-2*}
	\frac{\w_i\l_e(\P)}
		{\a_e\Big(\Pot_e^{\coeff{e}}(\P) - \Pot_e^{\coeff{e}}(\P \setminus \{i\})\Big)} \in \left[\frac{1}{\coeff{e}},\max\left\{1,\frac{k_e}{\coeff{e}}\right\}\right].
\end{equation}
%Part (a) of the theorem follows by Lemma~\ref{thm:fundamental} and the facts that $\gamma^*=\max_{e\in E}{\coeff{e}}$ and $d=\max_{e\in E}{k_e}$. Specifically, for $\coeff{e}=1$, we get that the quantity \eqref{thm:approx:ratio} is between $1$ and the expression in \eqref{eq:better-bound}. Then, part (b) of the theorem follows again by Lemma~\ref{thm:fundamental}.

Part (a) follows by combining Lemma~\ref{thm:fundamental}, \eqref{thm:approx:res-2}, \eqref{thm:approx:res-2*} and the fact that $\gamma^* = 1$. %; moreover, by adding \eqref{thm:approx:res-2*}, we obtain also that $\rho \leq \d$  (recall that $d=\max_{e\in E}{k_e}$).
Part (b) follows from Lemma~\ref{thm:fundamental}, \eqref{thm:approx:res-2*} and the definition of $\gamma^*$. %the fact that $\gamma^*=\max_{e\in E}{\coeff{e}}$. 
%and $d=\max_{e\in E}{k_e}$. Specifically, for $\coeff{e}=1$, we get that the quantity \eqref{thm:approx:ratio} is between $1$ and the expression in \eqref{eq:better-bound}. \\
\end{proof}

%ic: remove table
%\begin{table}[h!]
% \centering
% \caption{Bounds on $\rho(\G)$ (defined in Theorem \ref{thm:approx}) and comparison with the results in \cite{Hansknecht14}.}
% \begin{tabular}{c | c c}
%\hline
% & Our results & \cite{Hansknecht14} \\
%\hline
%\hline
%$\WCG(1)$ & $1$ & $1$ \\
%$\WCG(2)$ & $4/3$ & $4/3$ \\
%$\WCG(3)$ & $1.785$  & $1.785$ \\
%$\WCG(4)$ & $2.326$ & $2.326$ \\
%$\vdots$ &  &  \\
%$\WCG(\d)$ & {$\mathbf{\d}$} & {$\mathbf{\d+1}$} \\
%\hline
% \end{tabular}
% \label{tab:approx}
% \end{table}

We remark that $\rho$, as defined in part (a) of Theorem~\ref{thm:approx}, is considerably smaller than $\d$ for small values of the latter. In particular, we can show that $\rho=4/3$, $1.7848$, and $2.326$ for $\d=2$, $3$, and $4$, respectively. Interestingly, these values coincide with those obtained in \cite{Hansknecht14}, even though the expression that gives the approximation bound therein is different than ours. With our expression for $\rho$, we are able to bound it by $\d$ (as opposed to $\d+1$ in \cite{Hansknecht14}). Theorem~\ref{thm:approx} can also be used to obtain the next statement that has originally been proved in \cite{Christodoulou19}, as well as new statements such as Corollary~\ref{cor:congested} below and, more importantly, Theorem~\ref{thm:pos} in the next section.

\begin{corollary}\label{cor:opt}
	%Let  $\G = \game{\N, \E, (\w_i)_{i\in \N}, (\S_i)_{i\in \N}, (\a_e, \k_e)_{e\in \E}}$ be an instance of  \WCGD. 
	Any social optimum is a $(\d+1)$-approximate pure Nash equilibrium.
\end{corollary}
\begin{proof}
Let $\POT_\gamma$ be the state-valued function with $\gamma=(\coeff{e})_{e\in \E}$ and $\coeff{e} = \k_e + 1$ for $e\in E$.
	%Let us consider the function $\POT_\gamma(\ss)$ = \sum_{e\in \E}\a_e \Pot_e^{\coeff{e}}\big(\L_e(\ss)\big)$ (defined in Definition \ref{def:POT}), with $\coeff{e} = \k_e + 1$. 
	The claim follows by observing that $\POT_\gamma(\ss) = \C(\ss)$ and from the fact that, by Theorem \ref{thm:approx}, $\POT_\gamma$ is a $(\d+1)$-approximate potential function. 
\end{proof}

\begin{corollary}\label{cor:congested}
%Let $\tau>0$ and consider a $\tau$-congested game in \WCGD. Then, 
Let $\tau>0$, if the game is $\tau$-congested then the state-valued function $\POT_\gamma$ is an $\exp(1/\tau)$-approximate potential function.
\end{corollary}

\begin{proof}
The proof is along the lines of the proof of Theorem~\ref{thm:approx}. We remark that, even though the set $P$ is not restricted in the statement of Lemma~\ref{thm:fundamental}, whenever it is used in the proofs of Lemma~\ref{thm:fundamental} and Theorem~\ref{thm:approx}, it coincides with the set of players $\L_e(\ss)$ for a resource $e\in E$ and a state $\ss\in \SoS$. Then, the definition of $\tau$-congested game implies that the quantity $\mu_i(P)$, that is used in the proof of Theorem~\ref{thm:approx}, has value at least $\tau k_e$. Hence, the same proof of Theorem~\ref{thm:approx} yields that the state-valued function $\POT_\gamma$ is a $\rho$-approximate potential function with 
	\begin{equation}\label{eq:better-bound-tau}
\rho = \max_{e\in \E}\sup_{\x\geq \tau k_e}\; \frac{\big(1  + \x\big)^{\k_e}}
{\frac{1}{\k_e + 1}\big(1 + \x\big)^{\k_e + 1} + \frac{\k_e}{\k_e + 1} - \frac{1}{\k_e + 1}\x^{\k_e + 1}}.
\end{equation}

Due to the convexity of function $z^{k_e+1}$, the slope of the line connecting points $(x,x^{k_e+1})$ and $(1+x,(1+x)^{k_e+1})$, which is $(1+x)^{k_e+1}-x^{k_e+1}$, is at least as high as the value of the derivative of the function $z^{k_e+1}$ for $z=x$, i.e., $(k_e+1)x^{k_e}$. Hence, \eqref{eq:better-bound-tau} yields that
\begin{align*}
\rho \leq \max_{e\in \E}\sup_{\x\geq \tau k_e}\; \left(\frac{1  + \x}{\x}\right)^{\k_e} \leq \max_{e\in \E}\sup_{\x\geq \tau k_e}\; \exp(k_e/x)=\exp(1/\tau),
\end{align*}
and the theorem follows.
\end{proof}

%% file: pos.tex
%\subsection{Upper bound}
In this section we present our upper bound on the $\alpha$-approximate price of stability, for $\alpha\in [\d,\d+1]$. This bound is stated by Theorem \ref{thm:pos}; the proof uses the following lemma.

\begin{lemma}\label{lem:pos}
	Let $\delta\in [0,1]$ and $\gamma=(\coeff{e})_{e\in E}$ where $\coeff{e}=\min\{k_e+1,d+\delta\}$ for $e\in \E$. 
	For every state $\ss\in\SoS$, it holds that
	\begin{equation*}
	\POT_{\gamma}(\ss) \leq \C(\ss) \leq \frac{d+1}{d+\delta}\POT_{\gamma}(\ss).
	\end{equation*}
\end{lemma}

\begin{proof}
	Let $\E=\{e_1, e_2,\ldots, e_m\}$. 
	We need to bound the ratio 
	\begin{equation*}
	\frac{\C(\ss)}{\POT_{\gamma}(\ss)} = 
	\frac{
		\sum_{t=1}^{m}\a_{e_t}\left(\sum_{j\in \L_{e_t}(\ss)}\w_j\right)^{\k_{e_t} + 1}
	}
	{
		\sum_{t=1}^{m}\a_{e_t}\Pot_{e_t}^{\coeff{e_t}}\big(\L_{e_t}(\ss)\big)\phantom{^{\k_e+1}}
	}.
	\end{equation*}
	In order to do so, we consider the ratio between the $t$-th term in the numerator and the $t$-th term in the denominator, for every $t\in [m]$, that is
	\begin{equation}\label{lem:pos2:eq-1}
	\frac{
		\left(\sum_{j\in \L_{e_t}(\ss)}\w_j\right)^{\k_{e_t} + 1}
	}
	{
		\Pot_{e_t}^{\coeff{e_t}}\big(\L_{e_t}(\ss)\big)\phantom{^{\k_e+1}}
	}.
	\end{equation}
Recall the definition of the state-value function $\Pot_e^{\coeff{e}}$, which yields
\begin{align*}
\Pot_e^{\coeff{e}}(\L_{e_t}(\ss))=\frac{\coeff{e}}{\k_e + 1}\left(\sum_{j\in \L_{e_t}(\ss)}\w_j\right)^{k_e + 1} + \left(1 - \frac{\coeff{e}}{\k_e + 1}\right)\sum_{j\in \L_{e_t}(\ss)}\w_j^{\k_e + 1}.
\end{align*}	
By the definition of $\coeff{e}$ and since $\sum_{j\in \L_{e_t}(\ss)}\w_j^{\k_e + 1} \leq \left(\sum_{j\in \L_{e_t}(\ss)}\w_j\right)^{k_e + 1}$, we get that \eqref{lem:pos2:eq-1} is between $1$ and 
$\frac{k_e+1}{\coeff{e}}=\max\left\{1,\frac{k_e+1}{d+\delta}\right\}\leq \frac{d+1}{d+\delta}$. 
It follows that, $\C(\ss)/\POT_{\gamma}(\ss)$ is at least $1$ and at most $\frac{d+1}{d+\delta}$ and the lemma follows.
\end{proof}

\begin{theorem}\label{thm:pos}
	%Let  $\G = \game{\N, \E, (\w_i)_{i\in \N}, (\S_i)_{i\in \N}, (\a_e, \k_e)_{e\in \E}}$ be an instance of  \WCGD. 
	%Then 
	$\PoS{\d+\delta} \leq \frac{d+1}{d+\delta}$, for every $\delta \in [0,1]$. 
\end{theorem}

\begin{proof}
Let $\POT_\gamma$ be the function with $\gamma=(\coeff{e})_{e\in \E}$ and $\coeff{e} = \min\{\k_e + 1, d+\delta\}$ for $e\in E$.
Let $\oo\in \OPT$ be a social optimum.
Consider any sequence of $(\d+\delta)$-improvement moves starting from $\oo$.
By Theorem \ref{thm:approx}, we know that this sequence converges to a state which is a $(\d+\delta)$-approximate pure Nash equilibrium; we denote this state by $\ee$.
Moreover, along this sequence of moves, $\POT_\gamma$ is not increasing. Hence,
\begin{equation*}
	\POT_\gamma(\ee) \leq \POT_\gamma(\oo).
\end{equation*} 
Using this fact and applying Lemma \ref{lem:pos} repeatedly to both $\oo$ and $\ee$, we obtain 
\begin{equation*}
	\C(\ee) \leq \frac{d+1}{d+\delta}\POT_{\gamma}(\ee) \leq \frac{d+1}{d+\delta}\POT_{\gamma}(\oo) \leq \frac{d+1}{d+\delta}\C(\oo).
\end{equation*}
The theorem follows.
\end{proof}

%%%%____LOWER_BOUND____%%%% 	